\documentclass{ieeeaccess}
\usepackage{cite}
\usepackage{amsmath,amssymb,amsfonts}
\usepackage{graphicx}
\usepackage{textcomp}
\usepackage{color}
\newtheorem{theorem}{Theorem}
\newtheorem{lemma}{Lemma}
\newtheorem{definition}{Definition}

\newtheorem{proof}{Proof}

\newtheorem{example}{Example}
\newtheorem{remark}{Remark}

\usepackage{bm}
\usepackage{booktabs}
\usepackage{multirow}
\usepackage{algorithm}
\usepackage{algpseudocode}
\usepackage{subfigure}
\usepackage{amssymb}
\usepackage{cite}
\newcommand{\tabcaption}{\def\@captype{table}\caption}
\newcommand{\tabincell}[2]{\begin{tabular}{@{}#1@{}}#2\end{tabular}}

\def\BibTeX{{\rm B\kern-.05em{\sc i\kern-.025em b}\kern-.08em
    T\kern-.1667em\lower.7ex\hbox{E}\kern-.125emX}}
\begin{document}
\history{Date of publication xxxx 00, 0000, date of current version xxxx 00, 0000.}
\doi{10.1109/ACCESS.2017.DOI}

\title{Cascaded Coded Distributed Computing Schemes Based on Placement Delivery Arrays}
\author{\uppercase{Jing Jiang}\authorrefmark{1},\uppercase{Lingxiao Qu\authorrefmark{1}}}
\address[1]{Guangxi Key Lab of Multi-source Information Mining $\&$ Security, Guangxi Normal University, Guilin 541004, China (e-mail: jjiang2008@hotmail.com, lxqu1110@outlook.com).}



\corresp{Corresponding author: Jing Jiang (e-mail: jjiang2008@hotmail.com)}

\tfootnote{The work of J. Jiang was in part supported by  NSFC (Nos. 11601096, 61672176), 2016GXNSFCA380021,
Guangxi Higher Institutions Program of Introducing 100 High-Level Overseas Talents,
Guangxi Collaborative Innovation Center of Multi-source Information Integration and Intelligent Processing, the Guangxi ``Bagui Scholar" Teams for Innovation and Research Project,
and the Guangxi Talent Highland Project of Big Data Intelligence and Application.}

\begin{abstract}
Li {\it et al}. introduced coded distributed computing (CDC) scheme to reduce the communication load
in general distributed computing frameworks such as MapReduce.
They also proposed cascaded CDC schemes where each output function is computed multiple times, and proved that such schemes achieved the fundamental trade-off between computation load  and communication load.
However, these schemes require exponentially large numbers of input files
and output functions when the number of computing nodes gets large.
In this paper, by using the structure of placement delivery arrays (PDAs),
we construct several infinite classes of cascaded CDC schemes.
We also show that the numbers of output functions in all the new schemes are only
a factor of the number of computing nodes,
and the number of input files in our new schemes
is much smaller  than that of input files in CDC schemes derived by Li {\it et al}.
\end{abstract}


\titlepgskip=-15pt

\maketitle

\IEEEPARstart{W}ith the amount of data being generated increasing rapidly,
a computation task could have a huge amount of data to be processed.
Hence it becomes more and more difficult to complete such a task by using a single server.
Therefore, large scale distributed computing systems,
where a computation task is distributed to many servers, are becoming very relevant.
The MapReduce \cite{DG}, Hadoop \cite{AH} are two of the popular distributed computing frameworks,
and have wide application in many areas, for instance,
\cite{CH,CS,CSDWW,HZRS,LLPPR,LSR,LYMA,LYMA1,PLSSM,PRPA,YMA}.
In such frameworks,  in order to compute the output functions,
the computation consists of  three phases: map phase, shuffle phase and reduce phase.
In the map phase, distributed computing nodes process parts of the input data files locally,
generating some intermediate values according to their designed map functions.
In the shuffle phase, the nodes exchange the calculated intermediate values among each other,
in order to obtain enough values to calculate the final output results by using their designed reduce functions. In the reduce phase, each node computes its designed reduce functions by
using the intermediate values derived in the map phase and the shuffle phase.

In these frameworks, one has to spend  most of execution time on data shuffling.
For example, 33$\%$ of the execution time is spent on data shuffling in a Facebook's Hadoop cluster \cite{CZMJS}.
70$\%$ of the execution time is spent on data shuffling
when running ``SelfJoin'' on the Amazon EC2 cluster \cite{A}.
Coded distributed computing (CDC) introduced in \cite{LMYA} is an efficient approach
to reduce the communication load in shuffle phase.
The authors in \cite{LMYA} characterized a fundamental tradeoff between
``computation load" in the map phase and ``communication load"  in the shuffle phase
on homogeneous computing networks, i.e.,
all the nodes have the same storage, computing and communication capabilities.
All the schemes mentioned in this paper are based on homogeneous networks, unless otherwise specified.
There are many studies on CDC.
Some new CDC schemes are constructed by using combinatorial design in \cite{KR1,WCJ}.
In order to solve the straggler problem,
\cite{LLPPR,LAA,YWYT} proposed CDC schemes by using error-correcting codes.
The authors in \cite{SLC,WCJ1,WCJ2,XT} investigated CDC schemes in heterogeneous networks,
where nodes could have different storage, computation capabilities and communication constraints.

Observe that in \cite{LMYA}, the author  proposed CDC schemes
with $N={K\choose r}$ input files  and $Q={K\choose s}$ output functions,
where $K$ is the total number of computing nodes,
$r$ is the average number of nodes that map each file,
and $s$ is the number of nodes that compute each reduce function.
Obviously, the number of input files $N={K\choose r}$ and
the number of output functions $Q={K\choose s}$
increase too quickly with $K$ to be used in practice when $K$ is large.
That is, in order to achieve the gain of communication load,
a large number of input files and a large number of output functions
are needed.
In practical scenarios,
these requirements could significantly reduce
the promised gains of the method \cite{KR1}.
It is desired to design CDC schemes with smaller numbers of input files and output functions.
There are a few works paying attention to reducing the values of $N$ and $Q$,
for instance,  \cite{KR1}, \cite{WCJ},  \cite{YSW}.
However, the number $s$ of nodes that compute each reduce function
in most of the known schemes  is a certain value,
for example, $s=1$ in the schemes  in \cite{KR1} and \cite{YSW},
$s=1$ or $s=t$ where $t | K$, i.e., $t$ is a factor of $K$, in the schemes in \cite{WCJ}.
We list these known CDC schemes in Table \ref{known-CDC}.
%

{\begin{table*}
\center
\caption{Some known CDC schemes}\label{known-CDC}
\resizebox{450pt}{50pt}{
  \begin{tabular}{|c|c|c|c|c|c|c|c|}
\hline
Schemes and Parameters &
       \tabincell{c} { Number of  \\ Nodes $K$ } &
      \tabincell{c} { Computation  \\ Load $r$ } &
      \tabincell{c} { Replication  \\ Factor $s$ } &
      \tabincell{c} { Number of  \\ Files $N$ } &
      \tabincell{c} { Number of Reduce  \\ Functions $Q$ }&
      \tabincell{c} { Communication  \\ Load $L$} \\ \hline
\tabincell{c} {\cite{LMYA}, $K, r, s \in \mathbb{N}^+$ \\ with $1 \leq r, s \leq K$ }
      & $K$   & $r$ &  $s$  & ${K \choose r}$ & ${K \choose s}$
      & $\sum\limits_{\l=\max\{r+1,s\}}^{\min\{r+s,K\}}
      \frac{l{K \choose \l}{l-2 \choose r-1}{r \choose l-s}}{r{K \choose r}{K \choose s}}$ \\ \hline
\tabincell{c}{\cite{KR1}, $K, t \in \mathbb{N}^+$ \\ with $t\geq 2$ and  $t | K$}
      &$K$ & $t$ & $1$  &$(\frac{K}{r})^{r-1}$
      & $K$ & $\frac{1}{r-1}(1-\frac{r}{K})$\\ \hline
\tabincell{c}{ \cite{YSW}, $K, t \in \mathbb{N}^+$ \\ with $t\geq 2$ and  $t | K$}
      & $K$ &$K-t$ & $1$ & $(\frac{K}{r}-1)(\frac{K}{r})^{r-1}$
      & $K$& $\frac{1}{r-1}(1-\frac{r}{K})$\\ \hline
\tabincell{c}{ \cite{WCJ}, $K, t \in \mathbb{N}^+$ \\ with $t\geq 2$ and  $t | K$}
      & $K$ &$t$ & $t$ & $(\frac{K}{r})^{r-1}$ & $(\frac{K}{r})^{r-1}$
      &  $\frac{r^{r}(K-r)}{K^{r}(r-1)}
         + \sum\limits_{\l=2}^{r} (\frac{r}{K})^{r+\l}
         {r \choose \l}{\frac{K}{r} \choose 2}^{\l}\frac{2^{\l}\l}{2\l-1}$ \\ \hline

\end{tabular} }
\end{table*}}

However, in practice, the reduce functions are desired to be computed by multiple nodes,
which allows for consecutive MapReduce procedures as the reduce function outputs can act
as the input files for the next  MapReduce procedures \cite{ZCFSS}.
Such a kind of CDC schemes is always called {\it cascaded} CDC scheme.

In this paper, we focus on cascaded CDC schemes
with smaller numbers of input files and output functions.
\begin{itemize}
  \item[1)] Based on placement delivery array (PDA),
   which was introduced to construct coded caching schemes in \cite{YCTC},
   we propose a construction of cascaded CDC schemes.
      That is, given a known PDA, one can obtain a class of cascaded CDC schemes.
      We list three classes of new schemes in Table \ref{New-CDCs}.
  \item[2)] The numbers of output functions in all the new schemes are only
        a factor of the number of computing nodes $K$,
        and the number of input files in our new schemes
        is exponentially smaller in $K$ than that of the scheme in \cite{LMYA}.
  \item[3)] From the construction in this paper, we can obtain the schemes in \cite{YSW},
  i.e., our new schemes include the schemes in \cite{YSW} as a special case.
  In addition,  our new schemes include some schemes in  \cite{KR1}, \cite{LMYA} and \cite{WCJ}
  as a special case.
\end{itemize}

{\begin{table*}
\center
\caption{New CDC schemes}\label{New-CDCs}
\resizebox{450pt}{45pt}{
  \begin{tabular}{|c|c|c|c|c|c|c|c|}
\hline
Schemes and Parameters &
       \tabincell{c} { Number of  \\ Nodes $K$ } &
      \tabincell{c} { Computation  \\ Load $r$ } &
      \tabincell{c} { Replication  \\ Factor $s$ } &
      \tabincell{c} { Number of  \\ Files $N$ } &
      \tabincell{c} { Number of Reduce  \\ Functions $Q$ }&
      \tabincell{c} { Communication  \\ Load $L$} \\ \hline
\tabincell{c} { Scheme 1, $K, r, s \in \mathbb{N}^+$ \\ with $1 \leq r, s \leq K$}
      & $K$   & $r$ &  $s$  & ${K \choose r}$ & $\frac{K}{\gcd{(K,s)}}$
     & $\frac{s}{r}(1-\frac{r}{K})$\\ \hline
\tabincell{c} {Scheme 2, $K, t, s \in \mathbb{N}^+$ \\ with $t\geq 2$, $t|K$ and $s \leq K$}
    & $K$ & $t$ & $s$ &$(\frac{K}{r})^{r-1}$ & $\frac{K}{\gcd{(K,s)}}$
    & $\frac{s}{r-1}(1-\frac{r}{K})$\\ \hline
\tabincell{c} { Scheme 3, $K, t, s \in \mathbb{N}^+$ \\ with $t\geq 2$, $t|K$ and $s \leq K$}
    & $K$ &$K-t$ & $s$ & $(\frac{K}{r}-1)(\frac{K}{r})^{r-1}$ & $\frac{K}{\gcd{(K,s)}}$
    & $\frac{s}{r-1}(1-\frac{r}{K})$\\ \hline
\end{tabular} }
\end{table*}}

The rest of this paper is organized as follows.
In Section \ref{preliminaries}, we formulate a  general distributed computing framework.
In Section \ref{New-CDC},
a construction  of cascaded CDC schemes is proposed.
In Section \ref{Perf},
we analyse the performance of our new schemes from the construction in Section \ref{New-CDC},.
Finally conclusion is drawn in Section \ref{conclusion}.

\section{Preliminaries}
\label{preliminaries}

In this section, we give a formulation of our problem.
In this system, there are $K$ distributed computing nodes $\mathcal{K} = \{0,1,\ldots,K-1\}$,
$N$ files $\mathcal{W}= \{w_0, w_1, \ldots, w_{N-1}\}$, 
where $w_n \in \mathbb{F}_{2^D}$ for any $n \in \{0,1,\ldots,N-1\}$, 
and $Q$ output functions $\mathcal{Q} = \{\phi_0, \phi_1, \ldots, \phi_{Q-1}\}$,
where $\phi_q: \mathbb{F}_{2^D}^N \rightarrow \mathbb{F}_{2^B}$ for any $q \in \{0,1,\ldots,Q-1\}$,
which maps all the files to a bit stream $u_q = \phi_q(w_0, w_1, \ldots, w_{N-1}) \in \mathbb{F}_{2^B}$.
The goal of node $k$ $(k \in \mathcal{K})$ is responsible for computing a subset of output functions,
denoted by a set $\mathcal{Q}_k \subseteq \mathcal{Q}$.

\Figure[t!](topskip=0pt, botskip=0pt, midskip=0pt)[width=3.3 in]{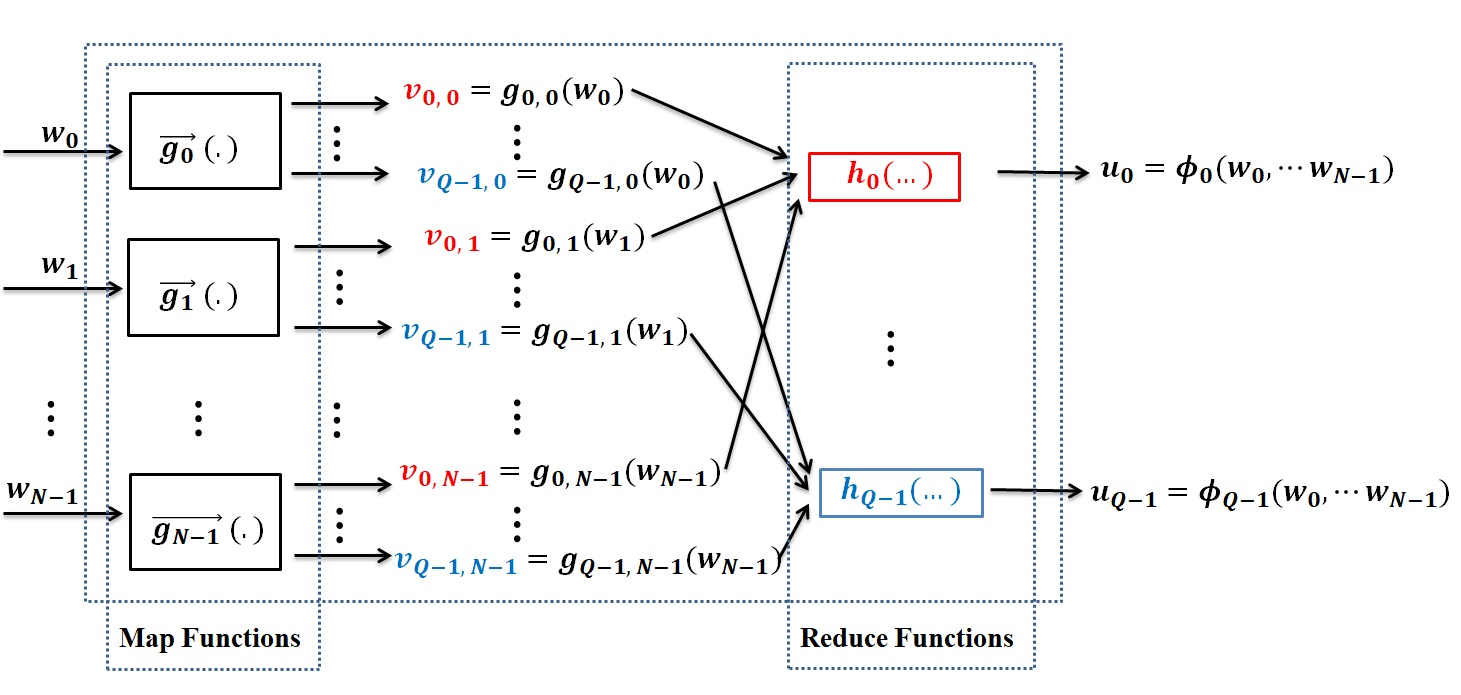}
{Illustration of a two-stage distributed computing framework \label{Fig.1}}

As illustrated in Fig. \ref{Fig.1} from  \cite{LMYA},
we assume that each output function $\phi_q$, $q \in \{0,1,\ldots,Q-1\}$, decomposes as:
$\phi_q(w_0, w_1, \ldots, w_{N-1})$ $=$
$h_q(g_{q,0}(w_0),$ $g_{q,1}(w_1),$ $\ldots,$ $g_{q,N-1}(w_{N-1}))$,
where

\begin{itemize}
\item[1)] The ``map'' function $g_{q,n}: \mathbb{F}_{2^D} \rightarrow \mathbb{F}_{2^T}$,
  $n \in \{0,$ $1,$ $\ldots,$ $N-1\}$, maps the file $w_n$ into the {\it intermediate value} (IVA)
  $v_{q,n}\triangleq g_{q,n}(w_n) \in \mathbb{F}_{2^T}$ for a given positive integer $T$.
\item[2)] The ``reduce'' function
   $h_{q}: \mathbb{F}_{2^T}^{N} \rightarrow \mathbb{F}_{2^B}$,
    maps the IVAs in $\{v_{q,n} \ | \ n\in \{0,1,\ldots,N-1\}\}$
    into the output value $u_q = h_{q}(v_{q,0}, v_{q,1}, \ldots, v_{q,N-1})$.
\end{itemize}

Following the above decomposition, the computation proceeds in the following three phases.

\begin{itemize}
  \item Map Phase. For each $k \in \mathcal{K}$,
  node $k$ computes  $g_{q,n}$ for $q \in \{0,1,\ldots,Q-1\}$ and $w_n\in \mathcal{W}_k$,
  where $\mathcal{W}_k \subseteq \mathcal{W}$ is the subset of files stored by $k$, i.e.,
  node $k$ locally computes a subset of IVAs
  \begin{equation*}\label{map-IVAs}
    \mathcal{C}_{k} = \{v_{q,n} \ | \ \phi_q \in \mathcal{Q}, w_n \in \mathcal{W}_k\}.
  \end{equation*}
  \item  Shuffle Phase. Denote $\mathcal{Q}_k$, $k \in \mathcal{K}$,
      as the subset of output functions which will be computed by node $k$.
      In order to obtain the output value of $\phi_q$ where $\phi_q \in \mathcal{Q}_k$,
      node $k$ needs to compute $h_{q}(v_{q,0}, v_{q,1}, \ldots, v_{q,N-1})$,
      i.e., it needs the IVAs that are not computed locally in the map phase.
      Hence, in this phase, the $K$ nodes exchange some of their computed IVAs.
      Suppose that node $k$ creates a message $X_k = \varphi_k(\mathcal{C}_{k})$
      of some length $l_k \in \mathbb{N}$, using a function
      $\varphi_k: \mathbb{F}_{2^T}^{|\mathcal{C}_{k}|} \rightarrow \mathbb{F}_{2^{\l_k}}$.
      Then it multicasts this message to all other nodes, where each node receives it error-free.
  \item Reduce Phase. Using the shuffled messages $X_0,$ $X_1,$ $\ldots,$ $X_{K-1}$
   and the IVAs in $\mathcal{C}_{k}$
   it computed locally in the map phase, node $k$ now could derive
   $(v_{q,0}, v_{q,1}, \ldots, v_{q,N-1}) = \psi_q(X_0, X_1, \ldots, X_{K-1}, \mathcal{C}_{k})$
   for some function
$\psi_q: \mathbb{F}_{2^{l_{0}}}\times\mathbb{F}_{2^{l_{1}}}\times\ldots\times \mathbb{F}_{2^{l_{K-1}}} \times \mathbb{F}_{2^T}^{|\mathcal{C}_{k}|} \rightarrow \mathbb{F}_{2^{T}}^N$
where $\phi_q \in \mathcal{Q}_k$.
More specifically,  node $k$ could derive the following IVAs
\begin{equation*}\label{reduce-IVAs-01}
   \{v_{q,n} \ | \ \phi_q \in \mathcal{Q}_k, n \in \{0,1,\ldots, N-1\}\},
\end{equation*}
which is enough to compute output value
$u_q$ $= h_{q}$ $(v_{q,0},$ $v_{q,1},$ $\ldots,$ $v_{q,N-1})$.
\end{itemize}

Define the {\it computation load} as $r=\frac{\sum_{k=0}^{K-1}|\mathcal{W}_k|}{N}$ and
{\it communication load} as $L=\frac{\sum_{k=0}^{K-1}l_k}{QNT}$,
i.e., $r$ is the average number of nodes that map each file
and $L$ is the ratio of the total number of bits transmitted in shuffle phase to $QNT$.
Li {\it et al}. \cite{LMYA} gave the following optimal computation-communication function.
\begin{equation}
\label{Ali-Com-load}
  L^{*}(r,s)
  = \sum\limits_{l=\max\{r+1,s\}}^{\min\{r+s,K\}}
  \frac{{K-r \choose K-l}{r \choose l-s}}{{K \choose s}} \frac{l-r}{l-1}\\,
\end{equation}
where $K$ is the number of nodes, $r$ is the computation load and $s$ is the number of nodes that compute each reduce function.
Moreover, the authors proposed some schemes achieving the above
optimal computation-communication function.

\begin{lemma}({\cite{LMYA}})
\label{Li-scheme}
Suppose that $K$, $r$, and $s$ are positive integers.
Then there exists a CDC  scheme with $K$ nodes,
$N={K\choose r}$ files and $Q={K\choose s}$ output functions,
such that the  communication load is
\begin{equation*}
\label{eq-Li-scheme}
L_{Li}(r,s)
  = \sum\limits_{l=\max\{r+1,s\}}^{\min\{r+s,K\}}
  \frac{{K-r \choose K-l}{r \choose l-s}}{{K \choose s}} \frac{l-r}{l-1}\\,
  \end{equation*}
where $r$ is the computation load and $s$ is the number of nodes that compute each reduce function.
\end{lemma}

For convenience,  the above CDC  schemes are called Li-CDC schemes in this paper.
Obviously, in the Li-CDC schemes,
the numbers of input files $N={K\choose r}$ and output functions $Q={K\choose s}$
are too large to be of practical use.
In the next section, we will construct some classes of CDC schemes
with smaller numbers of input files and output functions.

\section{A new Construction of CDC schemes}
\label{New-CDC}

In this section, we will give a construction of cascaded CDC schemes by using
placement delivery array which was introduced by Yan {\it et al}. \cite{YCTC} to study coded caching schemes.

\begin{definition}(\cite{YCTC})
\label{def-PDA}
Suppose that $K,N, Z$ and $S$ are positive integers.
$\mathbf{P}=(p_{n,k})$, $ n \in \{0,1,\ldots,N-1\}, k \in \{0,$ $1,$ $\ldots,$ $K-1\}$,
is an $N\times K$ array
composed of a specific symbol $``*"$  and positive integers $0,1,\cdots, S-1$.
Then $\mathbf{P}$ is a $(K,N,Z,S)$ placement delivery array $($PDA for short$)$
if
\begin{enumerate}
  \item [1)] the symbol $``*"$ occurs exactly $Z$ times in each column;
  \item [2)] each integer appears at least once in the array;
  \item [3)] for any two distinct entries $p_{n_1,k_1}$ and $p_{n_2,k_2}$,    $p_{n_1,k_1}=p_{n_2,k_2}=u$ is an integer  only if
  \begin{enumerate}
     \item [a.] $n_1\ne n_2$, $k_1\ne k_2$, i.e., they lie in distinct rows and distinct columns; and
     \item [b.] $p_{n_1,k_2}=p_{n_2,k_1}=*$, i.e., the corresponding $2\times 2$  subarray formed by rows $n_1,n_2$ and columns $k_1,k_2$ must be of the following form
  \begin{eqnarray*}
    \left(\begin{array}{cc}
      u & *\\
      * & u
    \end{array}\right)~\textrm{or}~
    \left(\begin{array}{cc}
      * & u\\
      u & *
    \end{array}\right).
  \end{eqnarray*}
   \end{enumerate}
\end{enumerate}
\end{definition}

A $(K,N,Z,S)$ PDA $\mathbf{P}$ is $g$-{\it regular}, denoted as $g$-$(K,$ $N,$ $Z,$ $S)$ PDA,
if each integer in $\{0,1,\ldots,S-1\}$ occurs exactly $g$ times in $\mathbf{P}$.

\begin{example}
\label{exa-pda}
We can directly check that the following array is a $3$-$(6,4,2,4)$ PDA:
\begin{equation*}
\mathbf{P}_{4\times 6}=
\bordermatrix{
& 0& 1& 2& 3 & 4 &5 \cr
0& *&*&0&*&1&2 \cr
1 & *&0&*&1&*&3 \cr
2& 0&*&*&2 &3&*\cr
3& 1 &2&3&*&*&* \cr
}.
\end{equation*}
\end{example}

The concept of a PDA is a useful tool to construct coded caching schemes,
for instance,  \cite{CJWY}, \cite{CJYT}, \cite{SZG}, \cite{YCTC}.
Recently, Yan {\it et al}.\cite{YSW}, \cite{YYW} used PDAs to construct CDC schemes
with $s=1$, where $s$ is the number of nodes that compute each reduce function.

\begin{lemma}(\cite{YSW})
\label{lem-PDA-CDC-same}
Suppose that there exists a $g$-$(K,N,Z,S)$ PDA with $g\geq 2$.
Then there exists a CDC scheme satisfying the following properties:
\begin{itemize}
 \item[1)] it consists of $K$ distributed computing nodes $\mathcal{K}$ $=$ $\{0,$ $1,$ $\ldots,$ $K-1\}$,
     $N$ files and $Q=K$ output functions $\mathcal{Q} = \{\phi_0, \phi_1, \ldots, \phi_{Q-1}\}$;
 \item[2)] node $k$, where $k \in \mathcal{K} $, is responsible for computing $\phi_{k}$;
 \item[3)] the computation load is $r=\frac{KZ}{N}$
 and the number of IVAs multicasted by all the nodes is $\frac{gS}{g-1}$.
\end{itemize}
\end{lemma}

In the above scheme, the numbers of nodes and files  are corresponding to the numbers of columns and rows of the PDA, respectively.
Furthermore, node $k$, $k\in\{0,1,\ldots,K-1\}$, is responsible for computing $\phi_{k}$,
i.e., distinct nodes are responsible for computing distinct output functions.
So the number $Q$ of output functions and the number of nodes are the same, i.e., $Q=K$.
In order to obtain the main results of this section,
the property that  some nodes are responsible for computing the same output function is needed.

\begin{example}
\label{exa-pda-CDC}
Consider the $3$-$(6,4,2,4)$ PDA $\mathbf{P}_{4\times 6}$ in Example \ref{exa-pda}.
By using $\mathbf{P}_{4\times 6}$, we will construct a CDC scheme with
$K=6$ nodes $\{0,1,2,3,4,5\}$, $N=4$ files $\{w_0,w_1,w_2,w_3\}$,
$Q=2$ output functions $\{\phi_0,\phi_1\}$
(note that $Q=K=6$ in Lemma \ref{lem-PDA-CDC-same}).

\begin{itemize}
  \item Map phase. For each $k \in \{0,1,2,3,4,5\}$, node $k$ stores the files in
\begin{equation}
\label{eq-ex-map-files}
\mathcal{W}_k=\{w_n \ | \ n \in\{0,1,2,3\},p_{n,k}=*\ \},
\end{equation}
i.e.,

\begin{eqnarray*}
  \mathcal{W}_0=\{w_0,w_1\},  \mathcal{W}_1=\{w_0,w_2\}, \mathcal{W}_2=\{w_1,w_2\},  \\ \mathcal{W}_3=\{w_0,w_3\},  \mathcal{W}_4=\{w_1,w_3\}, \mathcal{W}_5=\{w_2,w_3\}.
\end{eqnarray*}
  \item Shuffle phase. For each $k \in \{0,1,3,4\}$, node $k$ is responsible for computing $\phi_0$,
  i.e., it needs to obtain output value $u_{0}=h_0(v_{0,0},v_{0,1},v_{0,2},v_{0,3})$.
  For each $k \in \{2,5\}$, node $k$ is responsible for computing $\phi_1$.
  That is, node $k$ is responsible for computing $\phi_{q_k}$, where
  $(q_0,q_1,q_2,q_3,q_4,q_5)=(0,0,1,0,0,1)$.
  We take node $0$ as an example,
i.e.,
in order to compute the output value $u_{0}$,  node $0$ should obtain all the IVAs in
$\{v_{0,n}$ $\ | \ $ $n$ $\in$ $\{0,1,2,3\}\}.$
Since node $0$ stores $w_0$ and $w_1$, it can locally compute the IVAs in
\begin{equation}
\label{eq-local}
\{v_{0,n} \ | \ n\in \{0,1\}\}.
\end{equation}
Since $p_{2,0}=0$, node $0$ can derive $v_{0,2}$ from
the following subarray $\mathbf{P}^{0}$ formed by the rows
the rows $n_1$, $n_2$, $n_3$ and columns $k_1$,$k_2$,$k_3$,
where $p_{n_1,k_1}=p_{n_2,k_2}=p_{n_3,k_3}=0$.
\begin{equation*}
\label{eq-trans}
\mathbf{P}^{0}=
\bordermatrix{
  & 0& 1& 2\cr
0& *&*&0 \cr
1 & *&0&* \cr
2& 0&*&* \cr
}
\end{equation*}
Observe $\mathbf{P}^{0}$.
From \eqref{eq-ex-map-files}, $p_{n,k}=*$ means that the node $k$ can locally compute $v_{q_k,n}$.
So, from $\mathbf{P}^{0}$, we know that
node $0, 1$ and $2$ do not know $v_{q_0,2}=v_{0,2}$, $v_{q_1,1}=v_{0,1}$
and $v_{q_2,0}=v_{1,0}$, respectively.
Divide $v_{0,2}$, $v_{0,1}$ and $v_{1,0}$ into $2$ disjoint segments, respectively, i.e.,
\begin{equation*}
v_{0,2}=(v_{0,2}^{(1)},v_{0,2}^{(2)}), v_{0,1}=(v_{0,1}^{(0)},v_{0,1}^{(2)}), v_{1,0}=(v_{1,0}^{(0)},v_{1,0}^{(1)}).
\end{equation*}
For a segment $v_{q,n}^{(k)}$, the superscript $k$ represents the node that can locally compute this IVA,
which implies the segment $v_{q,n}^{(k)}$ will be transmitted by node $k$.
That is
\begin{itemize}
  \item node $0$ multicasts the message $v_{0,1}^{(0)}\bigoplus v_{1,0}^{(0)}$ to nodes $1$ and $2$,
  \item node $1$ multicasts the message $v_{0,2}^{(1)}\bigoplus v_{1,0}^{(1)}$ to nodes $0$ and $2$.
  \item node $2$ multicasts the message $v_{0,2}^{(2)}\bigoplus v_{0,1}^{(2)}$ to nodes $0$ and $1$.
\end{itemize}
Since node $0$ can compute $v_{1,0}$ and $v_{0,1}$ locally,
it can derive $v_{0,2}^{(1)}$ and $v_{0,2}^{(2)}$
from the messages $v_{0,2}^{(1)}\bigoplus v_{1,0}^{(1)}$  multicasted by  node $1$
and $v_{0,2}^{(2)}\bigoplus v_{0,1}^{(2)}$ multicasted by node $2$, respectively.
This implies that node $0$ can obtain $v_{0,2} =(v_{0,2}^{(1)},v_{0,2}^{(2)})$.
By using similar method, node $0$ can obtain $v_{0,3}$.
Together with the known IVAs in \eqref{eq-local},
node $0$ can obtain all the IVAs in $\{v_{0,n} \ | \  n\in \{0,1,2,3\}\}.$
Similarly, node $k$, $k \in \{1,3,4\}$ can obtain all the IVAs in
$\{v_{0,n} \ | \  n \in \{0,1,2,3\}\}$
and  node $k$, $k \in \{2,5\}$ can obtain all the IVAs in
$\{v_{1,n} \ | \  n \in \{0,1,2,3\}\}$.

\item  Reduce phase. By using the IVAs derived in map phase and shuffle phase,
for each $k_1 \in \{0,1,3,4\}$ and $k_2 \in \{2,5\}$,
node $k_1$ and $k_2$ can compute output values $u_0$ and $u_1$, respectively.
Since each node stores $2$ files,
the computation load is $r=\frac{2\times 6}{N}=\frac{2\times 6}{4}=3$.
For each integer $u\in \{0,1,2,3\}$, the number of IVAs multicasted by  the nodes from
$\mathbf{P}^{u}$ is $\frac{1}{2}\times3$.
So the total number of IVAs multicasted by  all the nodes  is $4 \times \frac{1}{2}\times 3 = 6$.
\end{itemize}
\end{example}

\begin{lemma}
\label{lem-PDA-CDC}
Suppose that there exists a $g$-$(K,N,Z,S)$ PDA with $g\geq 2$.
Then there exists a CDC scheme satisfying the following properties:
\begin{itemize}
  \item[1)] it contains $K$ distributed computing nodes $\mathcal{K}$ $=$ $\{0,$ $1,$ $\ldots,$ $K-1\}$,
  $N$ files and $Q$ output functions $\mathcal{Q}=\{\phi_0,\phi_1,\ldots,\phi_{Q-1}\}$, where $Q\leq K$;
  \item[2)] node $k$ is responsible for computing $\phi_{q_{k}}$, where $\phi_{q_{k}} \in \mathcal{Q}$;
  \item[3)] the computation load is $r=\frac{KZ}{N}$ and the number of IVAs multicasted by all the nodes is $\frac{gS}{g-1}$.
\end{itemize}
\end{lemma}

We now pay our attention to the proof of Lemma \ref{lem-PDA-CDC}.
Given a  $g$-$(K,N,Z,S)$ PDA $\mathbf{P}=(p_{n,k})$, $n \in $ $\{0,$ $1,$ $\ldots,$ $N-1\}$,
$k\in \{0,1,\ldots,K-1\}$,
we can construct a CDC scheme with
$K$ nodes $\mathcal{K}=\{0,1,\ldots, K-1\}$,
$N$ files $\mathcal{W}= \{w_0, w_1, \ldots, w_{N-1}\}$
and $Q$ output functions $\mathcal{Q} = \{\phi_0, \phi_1, \ldots, \phi_{Q-1}\}$,
where node $k \in \mathcal{K}$ is responsible for computing $\phi_{q_{k}}$
such that $\phi_{q_{k}}\in \mathcal{Q}$ and $\cup_{k\in \mathcal{K}}\phi_{q_{k}}=\mathcal{Q}$.

\begin{itemize}
  \item Map phase. Node $k$, where $k \in \{0,1,\ldots,K-1\}$, stores the files in
\begin{equation}
\label{map-files-1}
\mathcal{W}_k=\{w_n \ | \ n \in \{0,1,\ldots,N-1\},p_{n,k}=*\}.
\end{equation}
Hence the node $k$ can compute the IVAs in the set
\begin{equation}\label{map-IVAs-01}
\bigcup\limits_{p_{n,k}=*}v_{q_k,n}.
\end{equation}

We can also obtain the computation load
\begin{eqnarray*}
\begin{split}
r =\frac{\sum_{k=0}^{K-1}|\mathcal{W}_k|}{N}=\frac{\sum_{k=0}^{K-1}Z}{N} = \frac{KZ}{N}.
\end{split}
 \end{eqnarray*}

\item Shuffle phase. Note that node $k$, where $k$ $\in $ $\{0,$ $1,$ $\ldots,$ $K-1\}$,
     is responsible for computing the output function $\phi_{q_k}$.
     According to the definition of $g$-regular,
     each integer $u \in \{0,1,\ldots, S-1\}$ occurs $g$ times in $\mathbf{P}$.
     Suppose that $p_{n_1,k_1}=p_{n_2,k_2}=\ldots = p_{n_g,k_g}=u.$
     From condition 3) of Definition \ref{def-PDA},
     the subarray of $\mathbf{P}^{u}$  formed by the rows $n_1, n_2, \ldots, n_g$ and columns
    $k_1, k_2, \ldots, k_g$ is of the following form:
\begin{equation}
\label{eq-shuffle-u}
\mathbf{P}^{u}=
\bordermatrix{
  & k_1& k_2& \cdots & k_g\cr
n_1& u&*&\cdots & *\cr
n_2 & *&u&\cdots & * \cr
\vdots& \vdots&\vdots&\vdots \cr
n_g& *&*&\cdots & u \cr
}.
\end{equation}
Suppose that node $k_j$, $j \in \{1,2,\ldots,g\}$, is responsible for computing $\phi_{q_{k_j}}$.
Divide $v_{q_{k_j},n_j}$ into $g-1$ segments, i.e.,
\begin{equation}
\label{eq-IVA-seg}
v_{q_{k_j},n_j}=(v_{q_{k_j},n_j}^{(k_1)}, \ldots, v_{q_{k_j},n_j}^{(k_{j-1})},
v_{q_{k_j},n_j}^{(k_{j+1})}, \ldots , v_{q_{k_j},n_j}^{(k_{g})}).
\end{equation}
The superscript $k_{l}$, $l \in \{1,2,\ldots,g\}$, of a segment means that
such a segment will be transmitted by $k_{l}$.
For any $l \in \{1,2,\ldots,g\}$, the node $k_{l}$ multicasts the following message:
\begin{equation}
\label{shuffle-IVAs-02}
\bigoplus_{t\in\{1,2,\ldots, g\}\setminus\{l\}}v_{q_{k_{t}},n_{t}}^{(k_{l})}.
\end{equation}

Hence the number of IVAs multicasted by the nodes $k_1, k_2, \ldots, k_g$ is $\frac{g}{g-1}$
for the integer $u$.
Since there are $S$ integers in $\{0,1,\ldots,S-1\}$,
the total number of IVAs multicasted by all the nodes are $\frac{gS}{g-1}$.

In order to show the correctness of the scheme in the reduce phase,
we now prove that for any $j$ $\in$ $\{1,$ $2,$ $\ldots,$ $g\}$,
the node $k_{j}$ will obtain the IVAs $v_{q_{k_{j}},n_j}$ from $\mathbf{P}^{u}$ in \eqref{eq-shuffle-u},
where $p_{n_j,k_j}=u$.

\begin{enumerate}

\item  Since $p_{n_{1},k_1}$ $=$ $p_{n_2,k_2}$ $=$ $\ldots$ $=$ $p_{n_g,k_g}$ $=$ $u$,
according to the definition of a PDA, for any $j$ $\in$ $\{1,$ $2,$ $\ldots,$ $g\}$,
we have $p_{n_{t},k_{j}}=*$ for any $t$ $\in$ $\{1,$ $2,$ $\ldots,$ $g\}$ $\setminus$ $\{j\}$.
Then node $k_{j}$ stores file $w_{n_t}$ from \eqref{map-files-1},
which implies that it can locally compute $v_{q_{k_{l}},n_t}$ for any $l \in \{1,2,\ldots,g\}$.
So $k_{j}$ can compute $v_{q_{k_{t}},n_t}$ for any $t \in \{1,2,\ldots,g\}\setminus\{j\}$.
\item  We take $k_1$ as an example, i.e., node $k_{1}$ will obtain the IVA $v_{q_{k_{1}},n_1}$.
According to \eqref{eq-IVA-seg}, it need segments $v_{q_{k_{1}},n_1}^{(k_{2})},$ $ v_{q_{k_{1}},n_1}^{(k_{3})}, $ $\dots,$ $v_{q_{k_{1}},n_1}^{(k_{g})}$.
For the segment $v_{q_{k_{1}},n_1}^{(k_{2})}$,
from \eqref{shuffle-IVAs-02}, node $k_{2}$ multicasts the message
$\bigoplus_{t\in\{1,3,4,\ldots,g\}}v_{q_{k_{t}},n_t}^{(k_{2})}$.
From 1), $k_1$ can compute $v_{q_{k_{t}},n_t}$ for any $t \in \{2,3,\ldots,g\}$,
which implies it can locally compute
$v_{q_{k_{t}},n_t}^{(k_{2})}$ for any $t \in \{3,4,\ldots,g\}$.
So the node $k_{1}$ can obtain the segment $v_{q_{k_{1}},n_1}^{(k_{2})}$
from the  message $\bigoplus_{t\in\{1,3,4,\ldots,g\}}v_{q_{k_{t}},n_t}^{(k_{2})}$
multicasted by $k_2$.
Similarly, the node $k_{1}$ can obtain the segment $v_{q_{k_{1}},n_1}^{(k_{l})}$
for any $l \in \{3,4,\ldots,g\}$
from the  message $\bigoplus_{t\in\{1,2,\ldots,g\} \setminus \{l\}}v_{q_{k_{t}},n_t}^{(k_{l})}$
multicasted by $k_{l}$.
Now the node $k_1$  recovers the IVA
\begin{equation*}
v_{q_{k_1},n_1}=(v_{q_{k_1},n_1}^{(k_2)}, v_{q_{k_1},n_1}^{(k_{3})}, \ldots, v_{q_{k_1},n_1}^{(k_{g})}).
\end{equation*}
Similarly, for any $j \in \{2,3,\ldots,g\}$,
node $k_{j}$ could recover $v_{q_{k_{j}},n_{j}}$ from $\mathbf{P}^{u}$ in \eqref{eq-shuffle-u}.
\end{enumerate}

\item Reduce phase:
Consider node $k$, where $k \in \{0,1,\ldots, K-1\}$.
Since the node $k$ is responsible for computing $\phi_{q_{k}}$,
it needs to know the IVAs in the set
\begin{eqnarray*}
\label{map-IVAs-04}
\begin{split}
   & \{v_{q_k,n} \ | \ n \in \{0,1,\ldots,N-1\}\} \\
   & \ \ \ \ \ \ \ \ \ \ \ \ \ \ \   = (\bigcup\limits_{p_{n,k}=*}v_{q_k,n}) \bigcup(\bigcup\limits_{p_{n,k}\neq *}v_{q_k,n}) .
\end{split}
\end{eqnarray*}
From \eqref{map-IVAs-01}, the node $k$ can locally compute $\bigcup\limits_{p_{n,k}=*}v_{q_k,n}$.
Hence it only needs to derive $\bigcup\limits_{p_{n,k}\neq *}v_{q_k,n}$.
For any $p_{n,k}\neq*$, there exists an integer $u \in \{0,1,\ldots,S-1\}$ such that $p_{n,k}=u$.
From the shuffle phase, the node $k$ can get  the IVA $v_{q_k,n}$
from $\mathbf{P}^{u}$ in \eqref{eq-shuffle-u}.
That is node $k$ can derive all the IVAs in $\bigcup\limits_{p_{n,k}\neq *}v_{q_k,n}$.

\end{itemize}

Next, we will use Lemma \ref{lem-PDA-CDC} to derive some cascaded CDC schemes
where  the parameter $s$ of such schemes is a positive integer such that  $s >1$.

\begin{theorem}
\label{th-PDA-CDC}
Suppose that there exists a $g$-$(K,N,Z,S)$ PDA with $g\geq 2$.
Then for any  positive integer $s$ with $s\leq K$,
there exists a cascaded CDC scheme consisting of
$K$ distributed computing nodes, $N$ files and $Q=\frac{K}{\gcd{(K,s)}}$ output functions
such that the computation load is $r=\frac{KZ}{N}$ and the communication load is
$L=\frac{gsS}{(g-1)KN}$.
\end{theorem}

\begin{remark}
The number $Q=\frac{K}{\gcd{(K,s)}}$ of output functions in the above new scheme is only a factor of the number of nodes $K$, which is much smaller than the number $Q_{Li}={K \choose s}$ of output functions in Li-CDC scheme.
\end{remark}

\begin{remark}
Applying Theorem \ref{th-PDA-CDC} with $s=1$, the scheme is the same
as the scheme in  Lemma \ref{lem-PDA-CDC-same}.
That is, the schemes in Theorem \ref{th-PDA-CDC} include the schemes in \cite{YSW} as a special case.
\end{remark}

The rest of the section is devoted to the proof of Theorem \ref{th-PDA-CDC}.
Given a  $g$-$(K,N,Z,S)$ PDA $\mathbf{P}=(p_{n,k})$, $n$ $\in$ $\{0,$ $1,$ $\ldots,$ $N-1\}$, $k\in \{0,1,\ldots,K-1\}$,
we can construct a CDC scheme with $K$ nodes $\mathcal{K}=\{0,1,\ldots, K-1\}$,
$N$ files $\mathcal{W}= \{w_0, w_1, \ldots, w_{N-1}\}$ and $Q=\frac{K}{\gcd{(K,s)}}$ output functions
$\mathcal{Q} = \{\phi_0, \phi_1, \ldots, \phi_{Q-1}\}$.

\begin{itemize}
  \item Map phase. Node $k$, where $k \in \{0,1,\ldots,K-1\}$, stores the files in
\begin{equation}
\label{map-files}
\mathcal{W}_k=\{w_n \ | \ n \in \{0,1,\ldots,N-1\},p_{n,k}=*\},
\end{equation}
Hence the node $k$ can compute the IVAs in the set $\bigcup\limits_{p_{n,k}=*}v_{q_k,n}$.
We can also obtain the computation load
\begin{eqnarray*}
\begin{split}
r =\frac{\sum_{k=0}^{K-1}|\mathcal{W}_k|}{N}=\frac{\sum_{k=0}^{K-1}Z}{N} = \frac{KZ}{N}.
\end{split}
 \end{eqnarray*}
\item Shuffle phase. Node $k$, where $k \in \{0,1,\ldots,K-1\}$, is responsible for computing a subset of output functions
\begin{equation}
\label{shu-func}
\mathcal{Q}_k=\{\phi_{<ke>_Q}, \phi_{<ke+1>_Q}, \ldots, \phi_{<(k+1)e-1>_Q}\},
\end{equation}
where $e=\frac{sQ}{K}$ and
$<a>_b$ is the least non-negative residue of $a$ modulo $b$ for any positive integers $a$ and $b$.
We can directly check that $|\mathcal{Q}_k|=e$ and each output function is computed exactly $s$ times.

We divide the processes into $e$ steps such that in the $i$th step,
$1 \leq i \leq e$, node $k \in \mathcal{K}$ is responsible for computing $\phi_{<ke+i>_Q}$.
Then in such a step, by using Lemma \ref{lem-PDA-CDC},
the node $k$ can derive enough numbers of IVAs for computing $\phi_{<ke+i>_Q}$
and the number of IVAs multicasted by all the  nodes is $\frac{gS}{g-1}$.
There are $e$ steps, thus the total number of IVAs multicasted by all the  nodes is $\frac{gSe}{g-1}$.
So the communication load is
\begin{equation*}
  L=\frac{\frac{gSe}{g-1}}{QN}=\frac{\frac{gS\frac{sQ}{K}}{g-1}}{QN}=\frac{gsS}{(g-1)KN}.
\end{equation*}

\end{itemize}

\begin{example}
\label{ex-pda-CDC}
The following array is a $4$-$(10,5,3,5)$ PDA.
\begin{equation*}
\mathbf{P}_{5\times 10}=
\bordermatrix{
  & 0& 1& 2& 3 & 4 &5 & 6 & 7 & 8 & 9\cr
0& *&*&*&*&*&*&0&1 & 2 & 3 \cr
1 & *&*&*&0&1&2 & * &*&*&4\cr
2 & *&0&1&*&*&3 & * &*&4&*\cr
3& 0&*&2&*&3 &*&* & 4 &* &*\cr
4& 1 &2&*&3&*&* & 4 &*& *&* \cr
}
\end{equation*}

From Theorem \ref{th-PDA-CDC},
for $s=4$, we can construct a CDC scheme with
$K=10$ nodes $\mathcal{K}=\{0,1,\ldots,9\}$,
$N=5$ files $\mathcal{W}=\{w_0,w_1,w_2,w_3,w_4\}$,
$Q=\frac{K}{\gcd(K,s)}=5$ functions $\mathcal{Q}=\{\phi_0,\phi_1,\phi_2,\phi_3,\phi_4\}$.
Then, according to \eqref{map-files} and \eqref{shu-func},
\begin{eqnarray*}
\begin{split}
  \mathcal{W}_0=\{w_0,w_1,w_2\},  \mathcal{W}_1=\{w_0,w_1,w_3\}, \\
  \mathcal{W}_2=\{w_0,w_1,w_4\},  \mathcal{W}_3=\{w_0,w_2,w_3\},  \\
  \mathcal{W}_4=\{w_0,w_2,w_4\},  \mathcal{W}_5=\{w_0,w_3,w_4\},  \\
  \mathcal{W}_6=\{w_1,w_2,w_3\},  \mathcal{W}_7=\{w_1,w_2,w_4\}, \\
  \mathcal{W}_8=\{w_1,w_3,w_4\},  \mathcal{W}_9=\{w_2,w_3,w_4\},
\end{split}
\end{eqnarray*}
$e=\frac{sQ}{K}=2$, and
\begin{eqnarray*}
\begin{split}
\mathcal{Q}_0=\{\phi_0,\phi_1\},
\mathcal{Q}_1=\{\phi_2,\phi_3\}, \\
\mathcal{Q}_2=\{\phi_4,\phi_0\},
\mathcal{Q}_3=\{\phi_1,\phi_2\}, \\
\mathcal{Q}_4=\{\phi_3,\phi_4\},
\mathcal{Q}_5=\{\phi_0,\phi_1\}, \\
\mathcal{Q}_6=\{\phi_2,\phi_3\},
\mathcal{Q}_7=\{\phi_4,\phi_0\},\\
\mathcal{Q}_8=\{\phi_1,\phi_2\},
\mathcal{Q}_9=\{\phi_3,\phi_4\}.
\end{split}
\end{eqnarray*}
We divide the processes into $e=2$ steps.
\begin{itemize}
  \item In the first step, node $k$, where $k \in \{0,1,\ldots,9\}$, is responsible for computing $\phi_{<ke>_Q}
  =\phi_{<2k>_5}$.
We list them in Table \ref{ta-1st-ste}.
By using Lemma \ref{lem-PDA-CDC}, node $k$, $k$ $\in$ $\{0,$ $1,$ $\ldots,$ $9\}$
can derive enough IVAs for computing
$\phi_{<2k>_5}$.
 \begin{table}[h]
\caption{The first step of shuffle phase}\label{ta-1st-ste}
\resizebox{230pt}{15pt}{
  \begin{tabular}{|c|c|c|c|c|c|c|c|c|c|c|c}
\hline
Node&  0& 1& 2& 3 & 4 &5 & 6 & 7 & 8 & 9 \\ \hline
output function & $\phi_0$   & $\phi_2$ &  $\phi_4$  & $\phi_1$ & $\phi_3$
     & $\phi_0$   & $\phi_2$ &  $\phi_4$  & $\phi_1$ & $\phi_3$\\ \hline
  \end{tabular}}\centering
\end{table}
\item In the second step, similar to the first step,
node $k$, $k \in \{0,1,\ldots,9\}$ can derive enough IVAs for computing $\phi_{<ke+1>_Q}=\phi_{<2k+1>_5}$.
We list them in  Table \ref{ta-2nd-ste}.
 \begin{table}[h]
\caption{The second step of shuffle phase}\label{ta-2nd-ste}
\resizebox{230pt}{15pt}{
  \begin{tabular}{|c|c|c|c|c|c|c|c|c|c|c|c}
\hline
Node&  0& 1& 2& 3 & 4 &5 & 6 & 7 & 8 & 9 \\ \hline
output function & $\phi_1$   & $\phi_3$ &  $\phi_0$  & $\phi_2$ & $\phi_4$
     & $\phi_1$   & $\phi_3$ &  $\phi_0$  & $\phi_2$ & $\phi_4$\\ \hline
  \end{tabular}}\centering
\end{table}
\end{itemize}
So node $k$, $k\in \{0,1,\ldots, K-1\}$ can compute the output functions in $\mathcal{Q}_k$.
The computation load is
$r=\frac{10 \times 3}{5}=6$.
The total number of IVAs multicasted by all the  nodes is
$\frac{1}{3} \times 4 \times 10 = \frac{40}{3}$.
\end{example}

\section{Performance}
\label{Perf}

From Theorem \ref{th-PDA-CDC}, we can directly obtain CDC schemes from known PDAs.
We list some known results on $g$-$(K,N,Z,S)$ PDAs in Table \ref{table-known-results},
where $t | K$ represents that $t$ is a factor of $K$.
The interested readers can be reffered to
\cite{CJWY}, \cite{CJYT}, \cite{SZG}, \cite{YCTC} for more known results about PDAs .

{\begin{table*}
\center
\caption{Some known results on PDAs }\label{table-known-results}
\begin{tabular}{|c|c|c|c|c|c|c|c|}
\hline
References and Parameters& $g$ & $K$ & $N$& $Z$ &$S$ \\
\hline
\tabincell{c} {\cite{MN}, $K,t \in \mathbb{N}^+$ \\ with $1 \leq t \leq K-1$ }
             & $t+1$  & $K$  & ${K \choose t}$& ${K-1 \choose t-1}$ & ${K \choose t+1}$\\ \hline
\tabincell{c} {\cite{YCTC}, $K, t \in \mathbb{N}^+$ \\ with $t \geq 2$ and $t | K$}
             &$t$ &$K$& $(\frac{K}{t})^{t-1}$
             & $(\frac{K}{t})^{t-2}$& $(\frac{K}{t})^{t}-(\frac{K}{t})^{t-1}$\\ \hline
\tabincell{c} {\cite{YCTC}, $K, t \in \mathbb{N}^+$ \\ with $t \geq 2$ and $t | K$}
             & $K-t$ &$K$ &$(\frac{K}{t}-1)(\frac{K}{t})^{t-1}$
             &$(\frac{K}{t}-1)^{2}(\frac{K}{t})^{t-2}$ &$(\frac{K}{t})^{t-1}$\\ \hline
\end{tabular} 
\end{table*}}

\subsection{The first new scheme}

Let $\mathbf{P}$ be a $(t+1)$-$(K, {K \choose t}, {K-1 \choose t-1}, {K \choose t+1})$ PDA from \cite{MN}
(the PDA in the second row of Table \ref{table-known-results}).
From Theorem \ref{th-PDA-CDC}, for any positive integer $s\leq K$,
one can obtain a CDC scheme, say Scheme 1,
with $K$ nodes, $N={K \choose t}$ files and $Q=\frac{K}{\gcd{(K,s)}}$ output functions,
where $s$ is corresponding to the number of nodes that compute each output function.
Furthermore, the computation load is
  \begin{equation*}
    r=\frac{KZ}{N}=\frac{K{K-1 \choose t-1}}{{K \choose t}}=t,
  \end{equation*}
  and the communication load is
\begin{eqnarray*}
\label{new1-Com-load}
\begin{split}
L_1(r,s)
& =\frac{sgS}{(g-1)KN}=\frac{s(t+1){K \choose t+1}}{((t+1)-1)K{K \choose t}}\\
& =\frac{s}{t}(1-\frac{t}{K})=\frac{s}{r}(1-\frac{r}{K}).\\
\end{split}
\end{eqnarray*}

Note that if $s \geq \frac{rK}{K-r}$, we have $L_1(r,s)=\frac{s}{r}(1-\frac{r}{K})\geq 1$.
In this case, the nodes do not need to transmit coded messages.
Instead, each IVA can be multicasted by one node which can compute the IVA locally.
By using this method, the communication load $L_1(r,s) = 1$.
So the communication load is $L_1(r,s)=\min\{\frac{s}{r}(1-\frac{r}{K}),1\}$.

\subsubsection{Optimality}

When $s=1$ and $1 \le r < K-1$, \eqref{Ali-Com-load} can be written as
\begin{eqnarray*}
\begin{split}
L^{*}(r,1)
& = \sum\limits_{l=\max\{r+1,1\}}^{\min\{r+1,K\}}
  \frac{{K-r \choose K-l}{r \choose l-1}}{{K \choose 1}} \frac{l-r}{l-1}\\
& = \sum\limits_{l=\max\{r+1\}}^{\min\{r+1\}}
  \frac{{K-r \choose K-l}{r \choose l-1}}{{K \choose 1}} \frac{l-r}{l-1}\\
& =  \frac{1}{r}(1-\frac{r}{K}).\\
\end{split}
\end{eqnarray*}
Obviously, the communication load  $L_1(r,1)=\frac{1}{r}(1-\frac{r}{K})$
of Scheme 1 achieves the optimal computation-communication trade-off.

\begin{remark}
In this case, one can directly check that our scheme is the same as the scheme with $s=1$ proposed in \cite{LMYA}.
\end{remark}

When $1 \leq s \leq K$ and $r=K-1$, from \eqref{Ali-Com-load}, we have
\begin{eqnarray*}
\small{
\begin{split}
L^{*}(K-1,s)
& =\sum\limits_{l=\max\{K-1+1,s\}}^{\min\{K-1+s,K\}}
\frac{{K-(K-1) \choose K-l}{K-1 \choose l-s}}{{K \choose s}} \frac{l-(K-1)}{l-1}\\
& =\sum\limits_{l=\max\{K\}}^{\min\{K\}}
\frac{{K-(K-1) \choose K-l}{K-1 \choose l-s}}{{K \choose s}} \frac{l-(K-1)}{l-1}\\
&=\frac{s}{K-1}(1-\frac{K-1}{K})\\
&=\frac{s}{r}(1-\frac{r}{K}).\\
\end{split}}
\end{eqnarray*}

So in this case, the communication load of Scheme 1 achieves the optimal computation-communication trade-off.

\begin{remark}
\label{rem-scheme1}
One can show that the number of output functions in Scheme 1 is
much smaller than that of output functions in Li-CDC scheme.
The number of output functions in Li-CDC scheme is $Q_{Li}={K \choose s}$,
while the number of output functions in Scheme 1 is $Q_1=\frac{K}{\gcd(K,s)}$.
For example, if $s=K/w$ where $w$ is a factor of K, then $Q_{Li}= {K \choose K/w}$ and
$Q_1=\frac{K}{\gcd(K,K/w)} =w$, respectively.
We list the cases $w=2$ and $w=3$ when $2 \leq K \leq 20$ in Table \ref{ta-scheme1-opti-w=2}
and in Table \ref{ta-scheme1-opti-w=3}, respectively.
\end{remark}

 \begin{table}[h]
 \caption{The numbers of output functions in Li-CDC scheme and Scheme 1
 when $K$ is even and $s=\frac{K}{2}$}
 \label{ta-scheme1-opti-w=2}
  \begin{tabular}{|c|c|c||c|c|c|c|c|c|}
\hline
       \tabincell{c} { Number of \\ Node $K$  } &
       \tabincell{c} {  $Q_{Li}={K \choose s}$}  &
       \tabincell{c} { $Q_1=\frac{K}{\gcd(K,s)}$  } &
       \tabincell{c} { $K$  } &
       \tabincell{c} {  $Q_{Li}$}&
       \tabincell{c} { $Q_1$  }  \\ \hline
 $2$   & $2$ &  $2$ & $12$ & $924$ &  $2$ \\ \hline
 $4$   & $6$ &   $2$ & $14$ & $3432$ &  $2$ \\ \hline
 $6$ & $20$ &  $2$&   $16$ & $12870$ &  $2$\\ \hline
 $8$ & $70$ &  $2$&   $18$ & $48620$ & $2$ \\ \hline
 $10$ & $252$ &  $2$&  $20$ &$184756$ & $2$  \\ \hline
\end{tabular}
\centering
\end{table}

\begin{table}[h]
 \caption{The numbers of reduce functions in Li-CDC scheme and Scheme 1
 when $3|K$ and $s=\frac{K}{3}$}
 \label{ta-scheme1-opti-w=3}
  \begin{tabular}{|c|c|c||c|c|c|c|c|c|}
\hline
       \tabincell{c} { Number of \\ Node $K$  } &
      \tabincell{c} {  $Q_{Li}={K \choose s}$}  &
      \tabincell{c} { $Q_1=\frac{K}{\gcd(K,s)}$  } &
       \tabincell{c} { $K$  } &
      \tabincell{c} {  $Q_{Li}$} &
       \tabincell{c} { $Q_1$  } \\ \hline
 $3$   & $3$ &  $3$ & $12$ & $495$ &  $3$ \\ \hline
 $6$   & $15$ &   $3$  & $15$ & $3003$ &  $3$\\ \hline
 $9$ & $84$ &  $3$ & $18$ & $18564$ &  $3$\\ \hline
\end{tabular}
\centering
\end{table}

\subsubsection{Comparison}

For the other values of $K$, $r$ and $s$, we conjecture that
$H_1(r,s) = \frac{L_1(r,s)}{L^{*}(r,s)}\leq 2$.
Unfortunately, the structure of the formula
$L^{*}(r,s) = \sum\limits_{l=\max\{r+1,s\}}^{\min\{r+s,K\}}
\frac{{K-r \choose K-l}{r \choose l-s}}{{K \choose s}} \frac{l-r}{l-1}$
is too complex to prove this conjecture.
However, we could find out some values of $K$, $r$ and $s$
satisfying that $H_1(r,s) = \frac{L_1(r,s)}{L^{*}(r,s)}\leq 2$.

\begin{theorem}
\label{th-scheme1-1000}
For any positive integers $K$, $r$ and $s$ with $r,s \leq K \leq 1000$,
$H_1(r,s) = \frac{L_1(r,s)}{L^{*}(r,s)}\leq 2$.
\end{theorem}
\begin{proof}
With the aid of a computer, one can find out that $H_1(r,s) \leq 2$ holds for
all the positive integers $K$, $r$ and $s$ with $r,s \leq K \leq 1000$.
Here we only list some cases in Table \ref{ta-H1}.
For the other cases, we omit it and the interested readers may contact the author for a copy.
 \begin{table}[h]
 \caption{The ratio of $L_1(r,s)$ to $L^{*}(r,s)$ with $K=16$}\label{ta-H1}
\resizebox{230pt}{50pt}{
  \begin{tabular}{|c|c|c|c|c|c|c|c|}
\hline
    \tabincell{c} { Computation  \\ Load $r$  } &
       \tabincell{c} { Replication  \\ Factor $s$  } &
      \tabincell{c} {  Communication \\ Load $L^{*}(r,s)$  } &
      \tabincell{c} { Communication \\ Load $L_1(r,s)$ } &
      \tabincell{c} { $H_1(r,s)= \frac{L_1(r,s)}{L^{*}(r,s)}$ } \\ \hline
    & $1$   & $0.2708$ &  $0.2708$  & $1.0000$ \\ \cline{2-5}
$3$ & $2$   & $0.4333$ &    $0.5417$& $1.2500$ \\ \cline{2-5}
    &  $3$ & $0.5388$ &  $0.8125$&$1.5086$\\ \hline
 & $3$   & $0.3293$ & $0.4125$   & $1.2528$ \\ \cline{2-5}
$5$ & $5$   &$0.4540$  &   $0.6875$ & $1.5143$ \\  \cline{2-5}
 & $7$ & $0.5406$ &  $0.9625$&$1.7804$ \\ \hline
 &  $5$   & $0.2555$ &   $0.3125$ & $1.2233$ \\ \cline{2-5}
$8$ &  $8$   & $0.3579$ &   $0.5000$ & $1.3971$ \\  \cline{2-5}
 &  $10$ & $0.4125$ &  $0.6250$&$1.5150$ \\\hline
\end{tabular}}
\end{table}
\end{proof}

\begin{remark}
For the parameters $K,r$ and $s$ satisfying the conditions in Theorem \ref{th-scheme1-1000},
the communication loads of Scheme 1 are slightly larger than
those of Li-CDC scheme.
However, similar to Remark \ref{rem-scheme1},
one can show that the number of output functions in Scheme 1 is
much smaller than that of output functions in Li-CDC scheme.
\end{remark}

We also prove there exist some values of $K$, $r$ and $s$ such that
$H_1(r,s) = \frac{L_1(r,s)}{L^{*}(r,s)}\leq 2$ by using theoretical analysis.

\begin{lemma}
\label{lem-sche-theo-analy}
For any positive integers $K$, $r$ and $s$ with $r \geq s$
and $K \geq \frac{3rs(7r-s+1)}{8(r-s+1)}$,
$H_1(r,s) = \frac{L_1(r,s)}{L^{*}(r,s)}\leq 2$.
\end{lemma}

The proof of Lemma \ref{lem-sche-theo-analy} could be found in Appendix A.

By using Lemma \ref{lem-sche-theo-analy},
we can provide a method to show $H_1(r,s) \leq 2$ for some positive integers $s$, $r$ and $K$.
More specifically, given positive integers $r$ and $s$, one can find out an integer $K(r,s)$,
such that  $H_1(r,s) \leq 2$ if $K \geq K(r,s)$.
For $K < K(r,s)$, with the aid of a computer, one may check that whether $H_1(r,s) \leq 2$ holds.
We take the following result as an example.

\begin{theorem}
\label{th-sche-theo-analy}
Suppose that $K$, $r$ and $s$ are positive integers.
If $s\leq r \leq 8$ and $r\leq K$, then $H_1(r,s)=\frac{L_1(r,s)}{L^{*}(r,s)}\leq 2$.
\end{theorem}
\begin{proof}
We divide the proof into two parts.
\begin{itemize}
  \item [1)] $(r,s) \neq (8,8)$.
  We take $(r,s)=(2,2)$ as an example.
  According to Lemma \ref{lem-sche-theo-analy},
  for any positive integer $K \geq \frac{3rs(7r-s+1)}{8(r-s+1)}=19.5$, $H_1(r,s)\leq 2$ holds.
  For any positive integer $K < 19.5$, we can obtain $H_1(r,s)\leq 2$ from Theorem \ref{th-scheme1-1000}.
  Similarly, we can prove $H_1(r,s)\leq 2$ holds for any other pairs $(r,s)$.
  \item [2)] $(r,s) =(8,8)$.
  According to Lemma \ref{lem-sche-theo-analy},
  for any positive integer $K \geq \frac{3rs(7r-s+1)}{8(r-s+1)}=1176$, $H_1(r,s)\leq 2$ holds.
  For any positive integer $K \leq 1000$, we can obtain $H_1(r,s)\leq 2$ from Theorem \ref{th-scheme1-1000}.
  For any positive integer $1000< K < 1176$, with the aid of a computer,
  one can show that $H_1(r,s) \leq 2$ holds.
\end{itemize}
This completes the proof.
\end{proof}

\begin{remark}
It is easy to see that for all the parameters satisfying the conditions in Theorem \ref{th-sche-theo-analy}, the communication loads of Scheme 1 are slightly larger than those of Li-CDC scheme,
while the number of output functions in Scheme 1 is much
smaller than that of output functions in Li-CDC scheme.
Here we only list some cases in Table \ref{ta-scheme1-Q},
where $Q_{Li}$ and $Q_1$ are the numbers of output functions in Li-CDC scheme and Scheme 1, respectively.
\end{remark}

\begin{table}[h]
 \caption{The numbers of output functions in Li-CDC scheme and Scheme 1} \label{ta-scheme1-Q}
\resizebox{230pt}{40pt}{
  \begin{tabular}{|c|c|c|c|c|c|c|c|}
\hline
    \tabincell{c} { Number of \\ Node $K$  } &
     \tabincell{c} { Computation \\ Load $r$  } &
       \tabincell{c} { Replication  \\ Factor $s$  } &
      \tabincell{c} {  $Q_{Li}={K \choose s}$  } &
      \tabincell{c} { $Q_1=\frac{K}{\gcd(K,s)}$ }  \\ \hline
     & $3$ &   $2$ & $120$ & $8$\\ \cline{2-5}
 $16$   & $5$ &   $4$ & $1820$ & $4$\\  \cline{2-5}
  & $8$ &  $6$& $8008$ & $8$ \\ \hline
    & $3$ &   $2$ & $190$ & $10$\\ \cline{2-5}
 $20$   & $5$ &   $4$ & $4845$ & $5$\\  \cline{2-5}
  & $8$ &  $6$& $38760$ & $10$ \\ \hline
\end{tabular}}
\end{table}

\subsection{The second new scheme }
Let $\mathbf{P}$ be a $t$-$(K, (\frac{K}{t})^{t-1}, (\frac{K}{t})^{t-2}, (\frac{K}{t})^{t}-(\frac{K}{t})^{t-1})$
PDA from \cite{YCTC} (the PDA in the third row of Table \ref{table-known-results}).
From Theorem \ref{th-PDA-CDC},
for any positive integer $s\leq K$,
one can obtain a CDC scheme, say Scheme 2,
with $K$ nodes, $N=(\frac{K}{t})^{t-1}$ files
and $Q=\frac{K}{\gcd{(K,s)}}$ output functions,
where $s$ is corresponding to the number of nodes that compute each output function.
Furthermore, the computation load is
\begin{equation*}
 r=\frac{KZ}{N}=\frac{K(\frac{K}{t})^{t-2}}{(\frac{K}{t})^{t-1}}=t,
\end{equation*}
and the communication load is
\begin{eqnarray*}
\label{new1-Com-load}
\begin{split}
L_2(r,s)
& =\frac{sgS}{(g-1)KN}=\frac{st((\frac{K}{t})^{t}-(\frac{K}{t})^{t-1})}{(t-1)K(\frac{K}{t})^{t-1}}\\
& =\frac{s(K-t)}{(t-1)K}=\frac{s}{r-1}(1-\frac{r}{K}).\\
\end{split}
\end{eqnarray*}

Similar to the communication load of Scheme 1, it is possible that $L_2=\frac{s}{r-1}(1-\frac{r}{K}) > 1$.
By using similar method, we could have $L_2(r,s)=\min\{\frac{s}{r-1}(1-\frac{r}{K}),1\}$.
Obviously, the communication load in Scheme 2 is slightly larger than the communication load in Scheme 1,
i.e., $L_2=\frac{s}{r-1}(1-\frac{r}{K}) > \frac{s}{r}(1-\frac{r}{K})=L_1$.
Hence, similar to Scheme 1, we can also prove the following results.

\begin{theorem}
\label{th-scheme2}
Suppose that $K\geq 5$, $r$ and $s$ are positive integers satisfying that
$K \geq s$ and $r \geq 2 $ is  a factor of $K$.
Then $H_2(r,s) = \frac{L_2(r,s)}{L^{*}(r,s)} \leq 2.1$ holds
if one of the following conditions is satisfied:
\begin{itemize}
  \item [1)] $K\leq 1000;$
  \item [2)] $r \geq s+2$ and $K \geq  \frac{(111r-15s-111)rs}{44r-40s-44}$;
  \item [3)]  $s+2\leq r \leq 10$.
\end{itemize}
\end{theorem}

The proof of Theorem \ref{th-scheme2} is included in Appendix B.

Furthermore,  we can show that the number of files in Scheme 2 is
much smaller than that of files in Li-CDC scheme.
To do this, we need the following lemma.

\begin{lemma}(\cite{YCTC})
\label{le-known-equ}
For fixed rational number $a\in (0,1)$,
let $K \in \mathbb{N}^{+}$ such that $aK \in \mathbb{N}^{+}$,
when $K \rightarrow  \infty$,
\begin{eqnarray*}
{K \choose aK} \sim \frac{e^{K(a\ln\frac{1}{a}+(1-a)\ln\frac{1}{1-a})}}{\sqrt{2\pi K(a-a^2)}}.
\end{eqnarray*}
\end{lemma}

From Lemma \ref{Li-scheme},
the number of files in Li-CDC scheme is $K \choose aK$,
where $aK=r$. So according to Lemma \ref{le-known-equ}, the number of files in Li-CDC scheme is
\begin{eqnarray*}
{K \choose aK} \sim \frac{e^{K(a\ln\frac{1}{a}+(1-a)\ln\frac{1}{1-a})}}{\sqrt{2\pi K(a-a^2)}}
\end{eqnarray*}
when $K \rightarrow  \infty$. On the other hand, the number of files in Scheme 2 is $(\frac{K}{r})^{r-1}$,
which is exponentially smaller in $K$ than the number of files in  Li-CDC scheme.

\begin{remark}
From Theorem \ref{th-scheme2}, the communication load of Scheme 2 is slightly larger than
that of Li-CDC scheme. However, the number $\frac{K}{\gcd(K,s)}$ of output functions in Scheme 2
is much smaller than the number ${K \choose s}$ of output functions in Li-CDC scheme.
Furthermore, according to the above discussions,
the number of files in Scheme 2 is exponentially smaller in $K$
than the number of files in  Li-CDC scheme.
\end{remark}

\subsection{The third new scheme }
Let $\mathbf{P}$  be a $(K-t)$-$(K,$ $(\frac{K}{t}-1)(\frac{K}{t})^{t-1},$ $
(\frac{K}{t}-1)^{2}(\frac{K}{t})^{t-2},$ $(\frac{K}{t})^{t-1})$ PDA from \cite{YCTC}
(the PDA in the forth row of Table \ref{table-known-results}).
From Theorem \ref{th-PDA-CDC},
for any positive integer $s\leq K$,
one can obtain a CDC scheme, say Scheme 3,
with $K$ nodes, $N=(\frac{K}{t}-1)(\frac{K}{t})^{t-1}$ files
and $Q=\frac{K}{\gcd{(K,s)}}$ output functions,
where $s$ is corresponding to the number of nodes that compute each reduce function.
Furthermore, the computation load is
  \begin{equation*}
    r=\frac{KZ}{N}=\frac{K(\frac{K}{t}-1)^{2}(\frac{K}{t})^{t-2}}{(\frac{K}{t}-1)(\frac{K}{t})^{t-1}}=K-t,
  \end{equation*}
  and the communication load is
 \begin{eqnarray*}
\label{new1-Com-load}
\begin{split}
L_3(r,s)
& =\frac{sgS}{(g-1)KN}=\frac{s(K-t)(\frac{K}{t})^{t-1}}{(K-t-1)K(\frac{K}{t}-1)(\frac{K}{t})^{t-1}}\\
& =\frac{st}{(K-t-1)K}=\frac{s(K-r)}{(r-1)K} =\frac{s}{r-1}(1-\frac{r}{K}).\\
\end{split}
\end{eqnarray*}

We note that the communication load of Scheme 3 is equal to that of Scheme 2, i.e.,
$L_3=\frac{s}{r-1}(1-\frac{r}{K})=L_2$, which implies that
$H_3(r,s) = \frac{L_3(r,s)}{L^{*}(r,s)} = \frac{L_2(r,s)}{L^{*}(r,s)}=H_2(r,s)$.
So we can prove the following results by using the same method as the proof of Theorem \ref{th-scheme2}.

\begin{theorem}
\label{th-scheme3}
Suppose that $K\geq 5$, $r$ and $s$ are positive integers satisfying that
$s\leq K$ , $r\leq K-2$ and $K-r$ is  a factor of $K$.
Then $H_3(r,s) = \frac{L_3(r,s)}{L^{*}(r,s)} \leq 2.1$ holds
if one of the following conditions is satisfied:
\begin{itemize}
  \item [1)] $K\leq 1000;$
  \item [2)] $r \geq s+2$ and  $K \geq  \frac{(111r-15s-111)rs}{44r-40s-44}$.
\end{itemize}
\end{theorem}

Since the proof of Theorem \ref{th-scheme3} is the same as the proof of Theorem \ref{th-scheme2}-1) and 2),
we omit it here. The difference between Theorem \ref{th-scheme2} and Theorem  \ref{th-scheme3} is the
range of the computation load  $r$. In Theorem \ref{th-scheme2}, $r$ is a factor of $K$,
while $K-r$ is a factor of $K$ in Theorem \ref{th-scheme3}.

\begin{remark}
Similar to the discussions of Scheme 2, the communication load of Scheme 3 is slightly larger than
that of Li-CDC scheme. However, the number of output functions in Scheme 3
is much smaller than the number of output functions in Li-CDC scheme,
and  the number of files in Scheme 3 is exponentially smaller in $K$
than the number of files in  Li-CDC scheme.
\end{remark}

\section{Conclusion}
\label{conclusion}
In this paper, we paid our attention to cascaded CDC schemes on homogeneous computing networks.
We showed that, in many cases, $\frac{L}{L_{Li}} \leq 2.1$,
where $L$ and $L_{Li}$ are the communication loads of our new scheme and the scheme
derived by Li {\it et al}., respectively.
Most importantly, the number of output functions in all the new schemes are only
a factor of the number of computing nodes,
and the number of input files in our new schemes
is much smaller  than that of input files in CDC schemes derived by Li {\it et al}.


%

%

\ifCLASSOPTIONcaptionsoff
  \newpage
\fi



%
\bibliographystyle{IEEEtran}

\section*{Appendix A: Proof of Lemma 4}

Recall that the positive integers $K$, $r$ and $s$ satisfy $r\geq s$ and $K \geq \frac{3rs(7r-s+1)}{8(r-s+1)}$.

Firstly, consider the communication $L^{*}(r,s)$ of Li-CDC scheme.
Since
\begin{eqnarray}
\label{eq-value-K}
\begin{split}
K
& \geq  \frac{3rs(7r-s+1)}{8(r-s+1)}= \frac{3rs}{8}(7+\frac{6s-6}{r-s+1})\\
& \geq \frac{21}{8}rs >r+s,\\
\end{split}
\end{eqnarray}
we have
\begin{small}
\begin{equation*}
\begin{split}\label{simp-opti-tradeoff}
L^{*}(r,s)
&= \sum\limits_{l=\max\{r+1,s\}}^{\min\{r+s,K\}}
  \frac{{K-r \choose l-r}{r \choose l-s}}{{K \choose s}} \frac{l-r}{l-1}\\
&= \frac{1}{{K \choose s}}\sum\limits_{l=\max\{r+1,s\}}^{r+s}
  {K-r \choose l-r}{r \choose l-s} \frac{l-r}{l-1}\\
&\geq \frac{1}{{K \choose s}}{K-r \choose s}{r \choose r} \frac{s}{r+s-1}\\
&= \frac{s(K-r)(K-r-1)\cdots(K-r-s+1)}{(r+s-1)K(K-1)\cdots(K-s+1)}.
\end{split}
\end{equation*}
\end{small}
Then
\begin{small}
\begin{eqnarray}
\label{eq-H-1}
\begin{split}
H_1(r,s)
& = \frac{L_1(r,s)}{L^{*}(r,s)} \\
& \leq  \frac{s(K-r)}{Kr}\frac{(r+s-1)K(K-1)\cdots(K-s+1)}{s(K-r)(K-r-1)\cdots(K-r-s+1)}\\
& = \frac{r+s-1}{r}\frac{(K-1)\cdots(K-(s-1))}{(K-(r+1))\cdots(K-(r+s-1))}.\\
\end{split}
\end{eqnarray}
\end{small}

In order to evaluate the above value of $H_1(r,s)$, the following lemma is needed.

\begin{lemma}
\label{lem-evalu}
Suppose that K, $a_1$, $a_2$, \ldots, $a_n$ are  positive integers with $K>a_i$, $1 \leq i \leq n$,
and $(K-a_1)(K-a_2)\cdots(K-a_n)=K^{n}-b_1K^{n-1}+b_2K^{n-2} + \ldots + (-1)^{n-1}b_{n-1}K + (-1)^{n}b_{n}$. For any $1 \leq h \leq n-1$,
\begin{itemize}
  \item [1)] $b_{h+1}\leq \frac{\sum_{1 \leq i \leq n}a_i}{h+1}b_h$;
  \item [2)] $b_{h+1}K^{n-h-1} \leq b_hK^{n-h}$ holds if $K\geq \frac{\sum_{1 \leq i \leq n}a_i}{h+1}$.
\end{itemize}

\end{lemma}
\begin{proof}
Firstly, we will prove the first result.
It is not difficult to know that $b_h=\sum_{1 \leq i_1<i_2<\cdots < i_h \leq n} a_{i_1}a_{i_2}\cdots a_{i_h}$
and $b_{h+1}=\sum_{1 \leq i_1<i_2<\cdots < i_{h+1}\leq n} a_{i_1}a_{i_2}\cdots a_{i_{h+1}}$.
Then
\begin{tiny}
\begin{eqnarray*}
\begin{split}
b_{h+1}
& =\sum_{1 \leq i_1<i_2<\cdots < i_{h+1}\leq n} a_{i_1}a_{i_2}\cdots a_{i_{h+1}}\\
& =\frac{1}{h+1}\sum_{1\leq i_1<i_2<\cdots<i_{h+1}\leq n} (h+1)a_{i_1}a_{i_2}\cdots a_{i_{h+1}} \\
& =\frac{1}{h+1} \sum_{1 \leq i_1<i_2<\cdots < i_{h}\leq n}a_{i_1}a_{i_2}\cdots a_{i_h}
    \sum_{i_j\in \{1,\dots,n\}\setminus\{i_1,i_2,\ldots,i_h\}}a_{i_j}\\
& =\frac{1}{h+1}\sum_{1 \leq i_1<i_2<\cdots < i_{h}\leq n}a_{i_1}a_{i_2}\cdots a_{i_h}
    ( \sum_{i_j\in \{1,\dots,n\}}a_{i_j}\\
& \ \ \ \ \ \ -\sum_{i_j\in\{i_1,i_2,\ldots,i_h\}}a_{i_j})\\
& \leq\frac{\sum_{1 \leq i \leq n}a_i}{h+1}b_h. \\
\end{split}
\end{eqnarray*}
\end{tiny}
That is $b_{h+1}\leq \frac{\sum_{1 \leq i \leq n}a_i}{h+1}b_h$.

If $K\geq \frac{\sum_{1 \leq i \leq n}a_i}{h+1}$, together with the above result, we have
\begin{eqnarray*}
\begin{split}
b_{h+1}K^{n-h-1}
& \leq  \frac{\sum_{1 \leq i \leq n}a_i}{h+1}b_h K^{n-h-1}\\
& = (b_hK^{n-h})\frac{\sum_{1 \leq i \leq n}a_{i}}{K(h+1)} \\
& \leq b_hK^{n-h}.
\end{split}
\end{eqnarray*}
This completes the proof.
\end{proof}

By using Lemma \ref{lem-evalu},
\begin{small}
\begin{eqnarray}
\label{eq-H-1-nume}
\begin{split}
 & (K-1)\cdots(K-(s-1)) & \\
 = &K^{s-1}-\sum_{1 \leq i_1 \leq s-1} i_1K^{s-2}+\sum_{1 \leq i_1<i_2 \leq s-1}i_{1}i_{2}K^{s-3} \\
 &+ \ldots + (-1)^{s-1}\sum_{1 \leq i_1<i_2< \ldots <i_{s-1} \leq s-1}i_{1}i_{2}\ldots i_{s-1}\\
 \leq & K^{s-1}-\sum_{1 \leq i_1 \leq s-1} i_1K^{s-2}+\sum_{1 \leq i_1<i_2 \leq s-1}i_{1}i_{2}K^{s-3} \\
 \leq  &K^{s-1}-\sum_{1 \leq i \leq s-1} iK^{s-2}
   +\frac{\sum_{1 \leq i \leq s-1}i}{2}\sum_{1 \leq j \leq s-1}jK^{s-3}\\
 =  &K^{s-1}-\frac{s(s-1)}{2}K^{s-2}+\frac{s^2(s-1)^2}{8}K^{s-3},\\
\end{split}
\end{eqnarray}
\end{small}where the second and the third  from last formula come from Lemma \ref{lem-evalu}-1)
and Lemma \ref{lem-evalu}-2), respectively.
Similarly, by using Lemma \ref{lem-evalu}-2),
\begin{eqnarray}
\label{eq-H-1-deno}
\begin{split}
& (K-(r+1))\cdots(K-(r+s-1)) \\
\geq &  K^{s-1}-\sum_{r+1 \leq i_1 \leq r+s-1} i_1K^{s-2} \\
= & K^{s-1}-\frac{(2r+s)(s-1)}{2}K^{s-2}.\\
\end{split}
\end{eqnarray}
So, according to \eqref{eq-H-1}, \eqref{eq-H-1-nume}, \eqref{eq-H-1-deno},
\begin{eqnarray}
\begin{small}
\label{eq-H-1-main}
\begin{split}
H_1(r,s)
& \leq \frac{r+s-1}{r}\frac{(K-1)\cdots(K-(s-1))}{(K-(r+1))\cdots(K-(r+s-1))}\\
& \leq \frac{r+s-1}{r} \frac{K^{s-1}-\frac{s(s-1)}{2}K^{s-2}+\frac{s^2(s-1)^2}{8}K^{s-3}}
  {K^{s-1}-\frac{(2r+s)(s-1)}{2} K^{s-2}}\\
& = \frac{r+s-1}{r} \frac{8K^{2}-4s(s-1)K+s^2(s-1)^2}{8K^{2}-4(2r+s)(s-1)K}\\
& = \frac{r+s-1}{r} (1+\frac{8r(s-1)K+s^2(s-1)^2}{8K^{2}-4(2r+s)(s-1)K}).
\end{split}
\end{small}
\end{eqnarray}
Since $r \geq s$, we have
\begin{eqnarray*}
\begin{split}
& 8r(s-1)K < 8rsK, \\
& -4(2r+s)(s-1)K > -4(2r+r)sK=-12rsK.
\end{split}
\end{eqnarray*}
Since $rs <K$ from \eqref{eq-value-K},
$$s^2(s-1)^2 \leq r^2s^2 < rsK.$$
Hence
\begin{eqnarray*}
\begin{split}
H_1(r,s)
& <  \frac{r+s-1}{r} (1+\frac{8rsK+rsK}{8K^2-12rsK})\\
& =  \frac{r+s-1}{r} (1+\frac{9rs}{8K-12rs})\\
& \leq \frac{r+s-1}{r} (1+\frac{9rs}{8\frac{3rs(7r-s+1)}{8(r-s+1)}-12rs})\\
& =2,\\
\end{split}
\end{eqnarray*}
where the second from last formula holds since $K \geq \frac{3rs(7r-s+1)}{8(r-s+1)}$.

This completes the proof.

\section*{Appendix B: Proof of Theorem 4}

With the aid of a computer, one can show that $H_2(r,s) \leq 2.1$ holds for
all the positive integers satisfying condition 1) of  Theorem \ref{th-scheme2}.

Similar to the processes obtaining \eqref{eq-H-1-main}, one can obtain that
\begin{eqnarray*}
\label{eq-H-2-main}
\begin{split}
H_2(r,s)
& \leq \frac{r+s-1}{r-1} (1+\frac{8r(s-1)K+s^2(s-1)^2}{8K^{2}-4(2r+s)(s-1)K}).
\end{split}
\end{eqnarray*}

Since $r\geq s+2$ and $K \geq  \frac{(111r-15s-111)rs}{44r-40s-44}$, we have $rs <K$.
Then
$$s^2(s-1)^2 \leq r^2s^2 < rsK.$$
We also obtain that
\begin{eqnarray*}
\begin{split}
& 8r(s-1)K < 8rsK, \\
& -4(2r+s)(s-1)K > -4(2r+r)sK=-12rsK.
\end{split}
\end{eqnarray*}
Hence
\begin{eqnarray*}
\begin{split}
H_1(r,s)
& <  \frac{r+s-1}{r-1} (1+\frac{8rsK+rsK}{8K^2-12rsK})\\
& =  \frac{r+s-1}{r-1} (1+\frac{9rs}{8K-12rs})\\
& < \frac{r+s-1}{r-1} (1+\frac{9rs}{8\frac{(111r-15s-111)rs}{44r-40s-44}-12rs})\\
& =2.1.
\end{split}
\end{eqnarray*}
So $H_2(r,s) \leq 2.1$ holds for all the positive integers satisfying condition 2) of Theorem \ref{th-scheme2}.

We now consider the condition 3) of Theorem \ref{th-scheme2}.
We take $(r,s)=(4,2)$ as an example.
According to Theorem \ref{th-scheme2}-2)
for any positive integer $K \geq \frac{(111r-15s-111)rs}{44r-40s-44}=46.62$,
$H_2(r,s)\leq 2.1$ holds.
For any positive integer $K \leq  46$,
we can obtain $H_2(r,s)\leq 2.1$ from Theorem \ref{th-scheme2}-1).
Similarly, we can prove $H_2(r,s)\leq 2.1$ holds for any other pairs $(r,s)$
satisfying condition 3) of Theorem \ref{th-scheme2}.

This completes the proof.

\EOD

\end{document}